\documentclass[a4paper,USenglish]{llncs}
\usepackage{microtype}
\usepackage{subcaption}
\usepackage{lineno}
\usepackage{thm-restate}
\declaretheorem[name=Observation]{observation}

\usepackage{amsmath}
\usepackage{amsfonts}

\usepackage{amssymb}

\usepackage{hyperref}
\usepackage[nameinlink,capitalise]{cleveref}
\usepackage[utf8]{inputenc}

\usepackage{graphicx}
\usepackage{xspace}
\usepackage{paralist}

\usepackage{complexity}

\usepackage{todonotes}

\Crefname{observation}{Observation}{Observations}
\Crefname{algorithm}{Algorithm}{Algorithms}
\Crefname{section}{Section}{Sections}
\Crefname{appendix}{Appendix}{Appendices}
\Crefname{observation}{Observation}{Observations}
\Crefname{lemma}{Lemma}{Lemmas}
\Crefname{claim}{Claim}{Claims}
\Crefname{figure}{Fig.}{Figs.}
\Crefname{figure}{Fig.}{Figs.}
\Crefname{enumi}{Property}{Properties}

\usepackage{paralist}
\usepackage{color}
\definecolor{realblue}{rgb}{0,0,1}
\definecolor{darkerblue}{rgb}{0.094,0.455,0.804}
\definecolor{darkblue}{rgb}{0.063,0.306,0.545}
\definecolor{red}{rgb}{0.627,0.117,0.156}
\definecolor{green}{rgb}{0,0.588,0.509}
\definecolor{orange}{rgb}{0.903,0.739,0.382}
\definecolor{realred}{rgb}{1,0,0}

\newcommand{\red}[1]{{{\textcolor{red}{#1}\xspace}}}

\newcommand{\ver}{arxiv}
\newcommand{\arxapp}[2]{\ifthenelse{\equal{\ver}{conf}}{#2}{#1}}

\let\doendproof\endproof
\renewcommand\endproof{~\hfill$\qed$\doendproof}

\title{\texorpdfstring{$2$}{2}-Layer \texorpdfstring{$k$}{k}-Planar Graphs}
\subtitle{Density, Crossing~Lemma, Relationships, and Pathwidth}

\author{Patrizio Angelini \inst{1}\orcidID{0000-0002-7602-1524}, Giordano Da Lozzo\inst{2}\orcidID{0000-0003-2396-5174}, Henry F\"orster\inst{3}\orcidID{0000-0002-1441-4189}, and Thomas Schneck\inst{3}\orcidID{0000-0003-4061-8844}}
\authorrunning{Angelini {\em et al.}}
\institute{
    $^1$ John Cabot University, Rome, Italy
	\href{mailto:pangelini@johncabot.edu}{pangelini@johncabot.edu}\\
	$^2$ Roma Tre University, Rome, Italy
	\href{mailto:giordano.dalozzo@uniroma3.it}{giordano.dalozzo@uniroma3.it}\\
	$^3$ University of T\"ubingen, T\"ubingen, Germany
	\href{mailto:foersth@informatik.uni-tuebingen.de,schneck@informatik.uni-tuebingen.de}{\{foersth,schneck\}@informatik.uni-tuebingen.de}
}

\begin{document}

\maketitle              

\begin{abstract}
The $2$-layer drawing model is a well-established paradigm to visualize bipartite graphs. Several beyond-planar graph classes have been studied under this model. Surprisingly, however, the fundamental class of $k$-planar graphs has been considered only for $k=1$ in this context. We provide several contributions that address
this gap in the literature. First, we show tight density bounds for the classes of $2$-layer $k$-planar graphs with $k\in\{2,3,4,5\}$. Based on these results, we provide a Crossing Lemma for $2$-layer $k$-planar graphs, which then implies a general density bound for $2$-layer $k$-planar graphs. We prove this bound to be almost optimal with a corresponding lower bound construction. Finally, we study relationships between $k$-planarity and $h$-quasiplanarity in the $2$-layer model and show that $2$-layer $k$-planar graphs have pathwidth~at~most~$k+1$.
\end{abstract}

%

\section{Introduction}\label{sec:intro}

Beyond-planarity is an active research area that studies graphs admitting drawings that avoid certain forbidden crossing configurations. Research on this subject has attracted considerable interest due to its theoretical appeal and due to the need of visualizing real-world non-planar graphs.
A great deal of attention has been captured by 
two important graph families. The \emph{$k$-planar} graphs, with $k \geq 1$, for which the forbidden configuration is an edge crossing more than~$k$ other edges, and the \emph{$h$-quasiplanar} graphs, with~$h \geq 3$, for which the forbidden configuration is a set of $h$ pairwise crossing edges. The study of these two families finds its origins in the 1960's~\cite{avital-66,MR0187232}, when the question arose about the \emph{density} of these graphs, that is, the maximum number of edges of graphs in these families. 


Many works have addressed this extremal graph theoretical question  and established upper bounds for $k$-planar and $h$-quasiplanar graphs for various values of $k$ and $h$. For small $k$ and $h$, these upper bounds have been proven to be tight by lower bound constructions achieving the corresponding density.
The most significant results include tight density bounds for $1$-planar graphs~\cite{PachT97} ($4n-8$ edges), $2$-planar graphs~\cite{PachT97} ($5n-10$ edges), $3$-planar graphs~\cite{DBLP:conf/gd/Bekos0R16,PachRTT06} ($5.5n-20$), and $4$-planar graphs~\cite{DBLP:journals/corr/Ackerman15} ($6n-12$).
For general $k$, the currently best upper bound is $3.81\sqrt{k}\,n$, which can be derived from the result of Ackerman~\cite{DBLP:journals/corr/Ackerman15} on $4$-planar graphs and from the renowned Crossing Lemma~\cite{AiZi98}. 
For $h$-quasiplanar graphs, despite considerable research efforts, a density upper bound that is linear in the number of vertices 
exists only for $h \leq 4$~\cite{Ackerman09,AckermanT07,AgarwalAPPS97,DBLP:conf/jcdcg/PachRT02}. In particular, a tight upper bound exists for simple $3$-quasiplanar (for short, \emph{quasiplanar}) graphs. Here, \emph{simple} means that any two edges meet in at most one point, which is either a common endvertex or an internal point. For general $h$, only super-linear upper bounds are known~\cite{DBLP:journals/jct/CapoyleasP92,DBLP:conf/compgeom/FoxP08,DBLP:journals/siamdm/FoxPS13,DBLP:journals/algorithmica/PachSS96,DBLP:journals/comgeo/SukW15,DBLP:journals/dcg/Valtr98} while a linear bound has been conjectured~\cite{DBLP:journals/algorithmica/PachSS96}.

These two families have also been studied from other perspectives. A notable relationship is that every simple $k$-planar graph is also simple $(k+1)$-quasiplanar~\cite{DBLP:journals/jctb/AngeliniBBLBDHL20}, for every $k \geq 2$. It is also known that every \emph{optimal} $3$-planar~graph, namely one with the maximum possible number of edges ($5.5n-20$), is also $3$-quasiplanar. This latter result follows from a characterization of the optimal $3$-planar graphs~\cite{DBLP:conf/compgeom/Bekos0R17}, which also exists for the optimal $1$- and $2$-planar graphs~\cite{DBLP:conf/compgeom/Bekos0R17,PachT97}. Note that these characterizations do not directly yield recognition algorithms; in fact, recognizing (non-optimal) $k$-planar graphs is NP-complete for every $k \geq 1$~\cite{DBLP:journals/jgt/KorzhikM13}. The complexity of recognizing $h$-quasiplanar graphs is still open for any $h \geq 3$.

Aside these two major families, we mention the \emph{fan-planar} graphs, in which no edge is crossed by two independent edges or by two adjacent edges from different directions~\cite{BekosCGHK14,DBLP:journals/jgaa/BinucciCDGKKMT17,BinucciGDMPST15,KaufmannU14}, and the \emph{RAC graphs}, in which the edges are poly-lines with few bends and crossings only happen at right angles~\cite{DBLP:journals/tcs/AngeliniBFK20,DidimoEL11,Didimo2013,EadesL13}. These and other graph classes have been also investigated with respect to their density, recognition, and relationship with other classes; see also the recent survey~\cite{DBLP:journals/csur/DidimoLM19}. 

Beyond-planar classes have also been studied under additional constraints on the placement of the vertices. In the \emph{outer model}~\cite{AuerBBGHNR16,BekosCGHK14,DBLP:conf/gd/ChaplickKLLW17,DehkordiEHN16,DBLP:journals/ipl/Didimo13,HongEKLSS15,HongN15} every vertex is incident to the unbounded region of the drawing, while in the \emph{$2$-layer model}~\cite{DBLP:journals/jgaa/BinucciCDGKKMT17,BinucciGDMPST15,GiacomoDEL14,DBLP:journals/ipl/Didimo13} the vertices lie on two horizontal lines and every edge is a $y$-monotone curve.
The latter model requires the graph to be bipartite, and the constraints on the placement of the vertices emphasize the bipartite structure. Beyond-planar bipartite graphs have also been considered in the general drawing model, without any additional restriction~\cite{DBLP:conf/isaac/AngeliniB0PU18}. We remark that the $2$-layer model lies at the core of the Sugiyama framework for general layered drawings~\cite{DBLP:series/sseke/Sugiyama02,DBLP:journals/tsmc/SugiyamaTT81}.

In~\cite{GiacomoDEL14}, it was shown that $2$-layer RAC graphs have at most $\frac{3}{2}n-2$ edges and that this bound is tight, exploiting a characterization which also leads to an efficient recognition algorithm. Later, Didimo~\cite{DBLP:journals/ipl/Didimo13} observed that $2$-layer $1$-planar graphs are $2$-layer RAC graphs, and that the optimal graphs in these two classes coincide. Thus, the tight bound of $\frac{3}{2}n-2$ edges extends to $2$-layer $1$-planar graphs. For $h$-quasiplanar graphs, Walczak~\cite{gdInvitedTalk} provided a density upper bound of $(h-1)(n-1)$ edges, following from the fact that \emph{convex bipartite geometric} $h$-quasiplanar graphs can be $(h-1)$-colored so that edges with the same color do not cross. For ($3$-)quasiplanar graphs, the $2n-2$ bound can be improved to $2n-4$ by observing that they are planar bipartite graphs. Since fan-planar graphs are also quasiplanar, this density bound holds for $2$-layer fan-planar graphs, as well. Further, this bound is tight for both classes, since the complete bipartite graph $K_{2,n}$ is $2$-layer fan-planar. Note that $2$-layer fan-planar graphs have been characterized~\cite{DBLP:journals/jgaa/BinucciCDGKKMT17} and can be recognized when the graph is biconnected~\cite{DBLP:journals/jgaa/BinucciCDGKKMT17} or a tree~\cite{DBLP:journals/corr/abs-2002-09597}.
Another property that has been investigated in the $2$-layer model is the pathwidth. Namely, $2$-layer fan-planar graphs have pathwidth $2$~\cite{DBLP:journals/corr/abs-2002-09597}, while $2$-layer graphs with at most $c$ crossings in total have pathwidth $2c+1$~\cite{Dujmovic2008}; note that both results can be extended to general layered graphs.

\paragraph{Our Contribution.} From the above discussion it is evident that, in the wide literature on the $2$-layer model, the study of the central class of $k$-planar graphs is completely missing, except for the special case $k = 1$. In this paper, we make several contributions towards filling this gap. We provide tight density bounds for $2$-layer $k$-planar graphs with $k \in \{2,3,4,5\}$ in \cref{sec:smallK}. Exploiting these bounds, we deduce a Crossing Lemma for $2$-layer graphs in \cref{sec:densityLargeK}. This implies a density upper bound for general values of $k$. We then show a lower bound construction that is within a factor of $1/1.84$ from the upper bound. Finally, in \cref{sec:prop}, we investigate two additional properties. First, we prove that $2$-layer $2$-planar graphs are $2$-layer quasiplanar, as in the case where the vertices are not restricted to two layers~\cite{DBLP:journals/jctb/AngeliniBBLBDHL20}. For larger $k$, we show a stronger relationship, namely, every $2$-layer $k$-planar~graph is $2$-layer $h$-quasiplanar for $h =\left\lceil\frac{2}{3}k+2\right\rceil$. Second, we demonstrate that $2$-layer $k$-planar graphs have pathwidth at most $k+1$, which is the first result of this type, since they may have a linear number of crossings and may not be fan-planar. \arxapp{}{Note that missing proofs and full proofs of proof sketches may be found in~\cite{arxivVersion}.}

\section{Preliminaries}\label{sec:pre}

\paragraph{The $2$-layer model.} 
A \emph{bipartite graph} $G=(U\dot{\cup}V,E)$ is a graph with vertex subsets $U$ and $V$, so that $E \subseteq U  \times V$.
A \emph{topological $2$-layer graph} is a bipartite graph drawn in the plane so that the vertices in $U$ and $V$ are mapped to distinct
points on two horizontal lines $L_u$ and $L_v$, respectively, and the edges are mapped to $y$-monotone Jordan arcs. A topological $2$-layer graph can be assumed to be simple, that is, no two adjacent edges cross each other, and every two independent edges cross each other at most once.

Let $G$ be a topological $2$-layer graph. We denote the vertices in $U$ and in $V$ as $u_1,\ldots,u_p$ and $v_1,\ldots,v_q$, respectively, in the order in which they appear in positive $x$-direction along $L_u$ and $L_v$.  We denote the number of vertices of $G$ by $n=p+q$ and the number of edges in $E$ by $m$. We call $G$ \emph{$k$-planar} if each edge is crossed at most $k$ times, and \emph{$h$-quasiplanar} if there is no set of $h$ pairwise crossing edges. Further, we say that a bipartite graph $G$ is \emph{$2$-layer $k$-planar} (\emph{$h$-quasiplanar}) if there exists a topological $2$-layer $k$-planar (resp.\ $h$-quasiplanar) graph whose underlying abstract graph is isomorphic to $G$.

The \emph{maximum number of edges} of a graph class $\cal C$ is a function $m_\mathcal{C}: \mathbb{N} \rightarrow \mathbb{N}$ such that \begin{inparaenum}[(i)]
\item every $n$-vertex graph in $\cal C$ has at most $m_\mathcal{C}(n)$ edges, and
\item for every $n$, there is an $n$-vertex graph in $\cal C$ with $m_\mathcal{C}(n)$ edges.
\end{inparaenum} 
The \emph{(maximum edge) density} of $\cal C$ is a function $d_\mathcal{C}: \mathbb{N} \rightarrow \mathbb{N}$ such that \begin{inparaenum}[(i)]
\item for every $n$, it holds that $d_\mathcal{C}(n) \geq m_\mathcal{C}(n)$, and
\item there are infinitely many values of $n$ such that $d_\mathcal{C}(n) = m_\mathcal{C}(n)$.
\end{inparaenum} 
We say that an $n$-vertex graph in $\cal C$ with $d_\mathcal{C}(n)$ edges is \emph{optimal}. 

Note that $2$-layer quasiplanar graphs are equivalent to the \emph{convex bipartite geometric quasiplanar graphs}, where vertices lie on a convex shape so that the two partition sets are well-separated~\cite{gdInvitedTalk}. Since these graphs are planar bipartite, as discussed in \cref{sec:intro}, and include $K_{2,n}$, their density can be established using the same argumentation as for convex bipartite geometric quasiplanar~graphs~in~\cite{gdInvitedTalk}:

\begin{theorem}\label{thm:quasiplanarDensity}
An $n$-vertex $2$-layer quasiplanar graph has at most $2n-4$ edges  for $n \geq 3$. Also, there exist infinitely many $2$-layer quasiplanar graphs with $n$ vertices and $2n-4$ edges.
\end{theorem}

\paragraph{Tree and path decomposition.} A \emph{tree decomposition} of a graph $G=(V,E)$ is a tree $T$ on vertices $B_1,\ldots,B_n$ called \emph{bags} such that the following properties hold:
\begin{inparaenum}[(P.1)]
\item\label{prop:treedecomp1} each bag $B_i$ is a subset of $V$,
\item\label{prop:treedecomp2} $V = \bigcup_{i=1}^n B_i$,
\item\label{prop:treedecomp3} for every edge $(u,v) \in E$, there exists a bag $B_i$ such that $u,v \in B_i$, and
\item\label{prop:treedecomp4} for every vertex $v$, the bags containing $v$ induce a connected subtree of $T$.
\end{inparaenum}
If $T$ is a path, we call $T$ a \emph{path decomposition}. The \emph{width} of a tree decomposition $T$ is the maximum cardinality of any of its bags minus one, i.e., $\text{width}(T)=\max_{i \in\{1,\ldots,n\}} (|B_i|-1)$. The \emph{treewidth} of a graph $G$ is the minimum width of any of its tree decompositions, whereas the \emph{pathwidth} of $G$ is the minimum width of any of its path decompositions.


\section{Tight Density Results For Small Values of \texorpdfstring{$k$}{k}}
\label{sec:smallK}

In this section, we establish the density of $2$-layer $k$-planar graphs for small values of $k$. 
We start with a preliminary observation, which follows from the fact that the density of $k$-planar graphs can be upper bounded by a linear function in $n$~\cite{DBLP:journals/corr/Ackerman15,PachT97} and that the density of $2$-layer $1$-planar graphs is lower bounded by $\frac{3}{2}n-2$~\cite{GiacomoDEL14}. This allows us to derive the following\arxapp{ (see \cref{app:smallK} for a proof)}{}:

\begin{restatable}{lemma}{rationalDensity}
\label{lem:rational-density}
	 For  $k \geq 1$, there exist positive rational numbers $a_k \geq \frac{3}{2}$ and $b_k \geq 0$ such that \begin{inparaenum}[(i)]
	\item every $n$-vertex $2$-layer $k$-planar graph has at most $a_kn-b_k$ edges for $n \geq n_k$ with $n_k$ a constant,  and
	\item there is a $2$-layer $k$-planar graph with $n$ vertices and exactly $a_kn-b_k$ edges for some $n>0$.
\end{inparaenum}
\end{restatable}

We then define a useful concept for the analysis of $2$-layer $k$-planar graphs:

\begin{definition}
Let $G$ be a topological $2$-layer $k$-planar graph
and let $G[i,j|x,y]$, with $1 \leq i \leq j \leq p$ and $1 \leq x \leq y \leq q$, be the topological subgraph of $G$
induced  by vertices $\{u_i,\ldots,u_j,\allowbreak v_x,\ldots,v_y\}$.
$G[i,j|x,y]$ is a \emph{brick} if it 
contains two distinct crossing-free edges, namely $(u_i,v_x)$ and $(u_j,v_y)$, that are also crossing-free in $G$.
\end{definition}

The smallest brick, called \emph{trivial}, contains one vertex of one partition set, say $u_i=u_j$, and two consecutive vertices of the second one, say $v_x$ and $v_y=v_{x+1}$. 

\begin{observation}
Every optimal topological $2$-layer $k$-planar graph contains planar edges $(u_1,v_1)$ and $(u_p,v_q)$, and hence at least one brick.
\end{observation}

Regarding the connectivity we observe the following. If a  topological $2$-layer $k$-planar graph $G$ is not connected, we can draw the connected components as consecutive bricks and connect two consecutive bricks with another edge. Hence, we conclude the following:

\begin{observation}
Every optimal topological $2$-layer $k$-planar graph is connected.
\end{observation}

Next, we establish a useful property of an optimal $2$-layer $k$-planar graph $G$. 

\begin{lemma}\label{lem:twoLayerGraphsNoDegreeOneAndNoBrick}
Let $G$ be an optimal topological $2$-layer $k$-planar graph with exactly $a_kn-b_k$ edges. Then $G$ contains no vertex of degree $1$ and no trivial brick.
\end{lemma}

\begin{proof} 
Assume that $G$ contains a degree-$1$ vertex $v$ and consider the graph $G'$ obtained from $G$ by removing $v$. This graph has $m'=m-1$ edges and $n' =n-1$ vertices. Then, $m'=a_kn-b_k-1 = a_k(n-1)-b_k+(a_k-1)$, which is larger than $a_k(n-1)-b_k$ since $a_k \geq \frac{3}{2}$, by \cref{lem:rational-density}; a contradiction.

Second, assume that $G$ contains a trivial brick $G[i,i|x,x+1]$. Then, consider the graph $G'$ obtained from $G$ by identifying vertices $v_x$ and $v_{x+1}$. Clearly $G'$ has $m'=m-1$ edges (edges $(u_i,v_x)$ and $(u_i,v_{x+1})$ coincide in $G'$)  and $n' =n-1$ vertices. This leads to the same contradiction as in the previous case.  \end{proof}

\subsection{$2$-Layer $2$-Planar Graphs} We start with an observation about \emph{maximal} topological $2$-layer $2$-planar graphs, that is, in which no edge may be inserted without violating $2$-planarity.

\begin{figure}[t]
\centering
\begin{subfigure}[b]{0.475\textwidth}
\centering
\includegraphics[width=\textwidth,page=9]{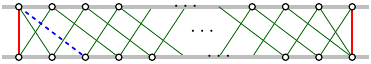}
\caption{}
\end{subfigure}
\hfil
\begin{subfigure}[b]{0.475\textwidth}
\centering
\includegraphics[width=\textwidth,page=36]{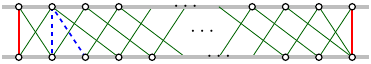}
\caption{}
\end{subfigure}
\caption{(a)~A maximal topological $2$-layer $2$-planar graph that is not optimal, as shown by the graph in~(b). Differences between the two graphs are dashed blue.}
\label{fig:accordion}
\end{figure}

\begin{observation}
There exists a maximal topological $2$-layer $2$-planar graph  that is not optimal; see \cref{fig:accordion}.
\end{observation}

We now characterize the structure of bricks in optimal $2$-layer $2$-planar graphs.

\begin{lemma}
\label{lem:bricksIn2planarGraphs}
Let $G$ be an optimal topological $2$-layer $2$-planar graph with exactly $a_2n-b_2$ edges and let $G[i,j|x,y]$ be a brick of $G$. Then, $j\geq i+1$ and $y = x+1$, or  $j = i+1$ and $y \geq x+1$.
\end{lemma}

\begin{proof}
By \cref{lem:twoLayerGraphsNoDegreeOneAndNoBrick}, $G[i,j|x,y]$  is not a trivial brick.
Assume, for a contradiction, that both $y \geq x+2$ and $j\geq i +2$. We first observe that $u_i$ is connected to some $v_t\neq v_x$, while $v_x$ is connected to some $u_s \neq u_i$. If this were not the case, say if $u_i$ were only incident to $v_x$, then a crossing-free edge $(v_x,u_{i+1})$ could be inserted, contradicting the optimality of $G$; see \cref{fig:2planarBricks1} and recall that a brick has no crossing-free edge, except for $(u_i,v_x)$ and $(u_j,v_y)$. So in the following assume that $(u_i,v_t)$ and $(v_x,u_s)$ belong to $G[i,j|x,y]$, with $v_t\neq v_x$ and $u_s \neq u_i$, such that there exists no edge $(u_i,v_{t'})$ with $t' > t$ and no edge $(v_x,v_{s'})$ with $s' > s$.
\begin{figure}[t]
\centering
\begin{subfigure}[b]{0.32\textwidth}
\centering \includegraphics[scale=1,page=31]{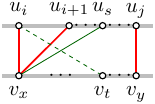}
\caption{}
\label{fig:2planarBricks1}
\end{subfigure}
\hfil
\begin{subfigure}[b]{0.32\textwidth}
\centering \includegraphics[scale=1,page=32]{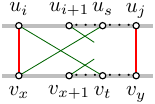}
\caption{}
\label{fig:2planarBricks2}
\end{subfigure}
\hfil
\begin{subfigure}[b]{0.32\textwidth}
\centering \includegraphics[scale=1,page=33]{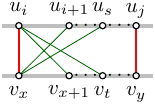}
\caption{}
\label{fig:2planarBricks3}
\end{subfigure}
\hfil
\begin{subfigure}[b]{0.525\textwidth}
\centering \includegraphics[width=\textwidth,page=34]{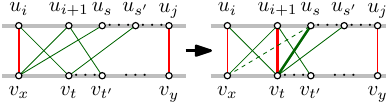}
\caption{}
\label{fig:2planarBricks4}
\end{subfigure}
\hfil
\begin{subfigure}[b]{0.45\textwidth}
\centering \includegraphics[width=\textwidth,page=35]{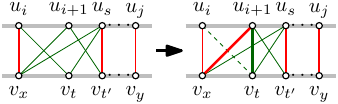}
\caption{}
\label{fig:2planarBricks5}
\end{subfigure}
\caption{Illustrations for the proof of \cref{lem:bricksIn2planarGraphs}.}
\label{fig:2planarBricks}
\end{figure}
Next, we consider $u_{i+1}$ and $v_{x+1}$. Assume first that $u_{i+1} \neq u_s$ and that $v_{x+1} \neq v_t$. Then, all edges incident to $u_{i+1}$ and $v_{x+1}$ have a crossing with  $(u_i,v_t)$ or $(v_x,u_s)$. Since $(u_i,v_t)$ and $(v_x,u_s)$ cross each other, there can be at most two such edges, and thus $u_{i+1}$ or $v_{x+1}$ has degree one; see \cref{fig:2planarBricks2,fig:2planarBricks3}. By \cref{lem:twoLayerGraphsNoDegreeOneAndNoBrick}, this contradicts the optimality of $G$. Hence, assume w.l.o.g. that $v_{x+1}=v_t$. Note that $u_s \neq u_{i+1}$, as otherwise the crossing-free edge $(u_{i+1},v_{x+1})$ could be inserted, contradicting the optimality of $G$. In addition, $u_s=u_{i+2}$, since otherwise $u_{i+1}$ and $u_{i+2}$ could only be incident to a total of two edges, by the same argument as above, resulting in a degree-$1$ vertex, which contradicts the optimality of $G$. 

By \cref{lem:twoLayerGraphsNoDegreeOneAndNoBrick}, both $u_{i+1}$ and  $v_{x+1}$ have degree at least $2$. Let $u_{s'}$ and $v_{t'}$ denote the neighbors of $v_{x+1}$ and  $u_{i+1}$ respectively, such that $s'$ and $t'$ are maximal. First assume that $t' \neq t$. If $s'=i+1$, the crossing-free edge $(u_s,v_t)$ can be inserted, contradicting the optimality of $G$. We observe that edge $(u_{i+1},v_{t'})$ is crossed by edges $(v_x,u_s)$ and $(v_t,u_{s'})$. If $u_s \neq u_s'$, we can obtain a topological $2$-layer $2$-planar graph $G'$ by removing edge  $(v_x,u_s)$  and inserting edges $(v_t,u_s)$ and $(v_t,u_{i+1})$; see \cref{fig:2planarBricks4}. This clearly contradicts the optimality of $G$. If $u_s = u_s'$, we can obtain a topological $2$-layer $2$-planar graph $G'$ by removing edge  $(u_i,v_t)$  and inserting edges $(v_x,u_{i+1})$ and $(v_t,u_{i+1})$; see \cref{fig:2planarBricks5}. This again contradicts the optimality of $G$. We conclude that $t'=t$.

Since $(v_x,u_s)$ is crossed by edges $(u_{i},v_{t})$ and $(u_{i+1},v_{t})$, we conclude that $(u_s,v_t)$ can be inserted without crossings, contradicting the optimality of $G$.
\end{proof}

By \cref{lem:twoLayerGraphsNoDegreeOneAndNoBrick,lem:bricksIn2planarGraphs}, we get that every brick must be a $K_{2,h}$ for some $h \geq 2$. The following observation shows that $h \leq 3$; see also \cref{fig:k_2-4}: 


\begin{observation}
\label{obs:noK33orK24}
The complete bipartite graph $K_{2,4}$  is not $2$-layer $2$-planar.
\end{observation}

%

We are ready to prove a tight bound for the density of $2$-layer $2$-planar graphs:

\begin{theorem}
\label{thm:2planarUpperbound}
Any $2$-layer $2$-planar graph on $n$ vertices has at most $\frac{5}{3}n-\frac{7}{3}$ edges. Moreover, the optimal $2$-layer $2$-planar graphs with exactly $\frac{5}{3}n-\frac{7}{3}$ edges are sequences of $K_{2,3}$'s such that consecutive $K_{2,3}$'s share one planar~edge.
\end{theorem}
\begin{figure}[t]
\centering
\begin{subfigure}[b]{0.2\textwidth}
\centering \includegraphics[scale=1,page=4]{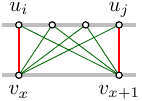}
\caption{}
\label{fig:k_2-4}
\end{subfigure}
\hfil
\begin{subfigure}[b]{0.2\textwidth}
\centering \includegraphics[scale=1,page=6]{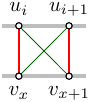}
\caption{}
\label{fig:k22piece}
\end{subfigure}
\hfil
\begin{subfigure}[b]{0.22\textwidth}
\centering \includegraphics[scale=1,page=7]{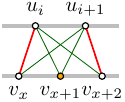}
\caption{}
\label{fig:k23piece}
\end{subfigure}
\hfil
\begin{subfigure}[b]{0.33\textwidth}
\centering \includegraphics[scale=1,page=8]{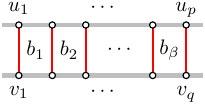}
\caption{}
\label{fig:pieceSequence}
\end{subfigure}
\caption{The unique $2$-layer drawings of 
(a)~$K_{2,4}$;  
(b)~$K_{2,2}$;
(c)~$K_{2,3}$.
(d)~An optimal $2$-layer $2$-planar graph is a sequence of bricks joint at planar edges.}
\label{fig:structureOfOptimal2Planar}
\end{figure} 
\begin{proof}
%
%

\cref{lem:twoLayerGraphsNoDegreeOneAndNoBrick,lem:bricksIn2planarGraphs}, and \cref{obs:noK33orK24} imply that $G$ contains only $K_{2,2}$- and $K_{2,3}$-bricks;
see \cref{fig:k22piece,fig:k23piece}.
Moreover, the planar edges separate $G$ into a sequence of $\beta$ bricks $(b_1,\ldots, b_\beta)$
such that $b_i$ and $b_{i+1}$ share one planar edge. Let $\beta_2$ denote the number of $K_{2,2}$-bricks.
Then, $G$ has $\beta-\beta_2$ $K_{2,3}$-bricks. Moreover, $n = 2\beta+2+(\beta-\beta_2) = 3\beta - \beta_2 + 2$
since each of the $\beta+1$ planar edges is incident to two distinct vertices
while each $K_{2,3}$-brick contains an additional vertex; see  \cref{fig:k23piece}.
Finally, $m=\beta+1+2\beta_2+4(\beta-\beta_2) = 5 \beta -2 \beta_2 + 1$ since every $K_{2,2}$-brick contains two non-planar edges
while every $K_{2,3}$-brick contains four. For a fixed value of $n$, $\beta=\frac{1}{3}n+\frac{1}{3}\beta_2-\frac{2}{3}$ and the density is $m=\frac{5}{3}n-\frac{1}{3}\beta_2-\frac{7}{3}$. This is clearly maximized for $\beta_2=0$. Hence, the maximum density is $m=\frac{5}{3}n-\frac{7}{3}$ which is tightly achieved for graphs in which every brick is a $K_{2,3}$. 
%
%
\end{proof}

\subsection{$2$-Layer $3$-Planar Graphs}  Next, we give a tight bound on the density of $2$-layer $3$-planar graphs. We first present a lower bound construction:

\begin{theorem}
There exist infinitely many $2$-layer $3$-planar graphs with $n$ vertices and $2n - 4$ edges.
\label{thm:3planarLowerbound}
\end{theorem}

\begin{proof}
\begin{figure}[t]
  \centering
  \includegraphics[scale=1,page=11]{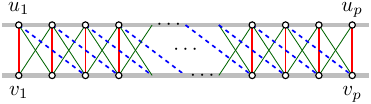}
  \caption{A family of $3$-planar graphs on $n = 2p$ vertices with $2n-4$ edges.}
  \label{fig:3planarLowerbound}
\end{figure}
We describe a family of graphs where $p=q$; refer to \cref{fig:3planarLowerbound}. Each graph has the following edges:
$(u_i,v_i)$ for $1 \leq i \leq p$ (red edges in \cref{fig:3planarLowerbound});
$(u_i,v_{i+1})$ for $1 \leq i \leq p-1$, and $(u_i,v_{i-1})$ for $2 \leq i \leq p$ (green edges in \cref{fig:3planarLowerbound});
$(u_i,v_{i+2})$ for $1 \leq i \leq p-2$ (dashed blue edges in \cref{fig:3planarLowerbound}).
Vertices $u_1$, $u_{p-1}$, $v_2$ and $v_p$ have degree 3, $u_{p}$ and $v_1$ have degree 2,
and all other vertices have degree 4, yielding $4n-8$ for the sum of the vertex degrees
and hence $2n-4$ edges.
\end{proof}

The following theorem provides the corresponding density upper bound:

\begin{restatable}{theorem}{threePlanarUpperbound}
\label{thm:3planarUpperbound}
Let $G$ be a topological $2$-layer $3$-planar graph on $n$ vertices. Then $G$ has at most $2n - 4$ edges for $n \geq 3$.
Moreover, if $G$ is optimal,  it is quasiplanar.
\end{restatable}

\begin{proof}[Sketch]
We show that optimal $2$-layer $3$-planar graphs are quasiplanar, which implies the statement, by \cref{thm:quasiplanarDensity}. \arxapp{Refer to \cref{app:smallK} for details.}{}
\end{proof}



\subsection{$2$-Layer $4$-Planar Graphs}

We first present a lower bound construction for this class of graphs:

\begin{theorem}
\label{thm:4planarLowerbound}
There exist infinitely many $2$-layer $4$-planar graphs with $n$ vertices and $2n - 3$ edges.
\end{theorem}

\begin{proof}
\begin{figure}[t]
\centering
  \begin{subfigure}[b]{0.6\textwidth}
\centering \includegraphics[scale=1,page=12]{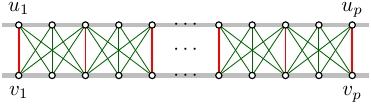}
\caption{}
\label{fig:4planarLowerbound}
\end{subfigure}
\hfil
  \begin{subfigure}[b]{0.35\textwidth}
\centering 
\includegraphics[scale=1,page=13]{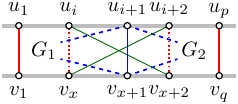}
\caption{}
\label{fig:4planarUpperbound}
\end{subfigure}
\caption{(a)~A family of $4$-planar graphs on $n=2p$ vertices with $2n-3$ edges. (b)~A triple of pairwise crossing edges and at most 4 additional edges separates an optimal $2$-layer $4$-planar graph into graphs $G_1$ and $G_2$.}
\label{fig:4planar}
\end{figure}
%
We describe a family of graphs where $p=q$;
see \cref{fig:4planarLowerbound}.
Each topological graph $G$ consists of a sequence $(b_1,\ldots,b_\beta)$ of $K_{3,3}$-bricks
such that $b_i$ and $b_{i+1}$ share a planar edge for $1 \leq i \leq \beta-1$.
Then $G$ has $n=4\beta+2$ vertices and $m=8\beta+1=2n-3$ edges.
\end{proof}

Next, we provide a matching upper bound.

\begin{restatable}{theorem}{fourPlanarUpperbound}
\label{thm:4planarUpperbound}
Any $2$-layer $4$-planar graph on $n$ vertices has at most $2n - 3$ edges.
\end{restatable}

\begin{proof}[Sketch]
%
%
We first prove that in an optimal topological $2$-layer $4$-planar graph $G$, every triple of pairwise crossing edges is such that removing the triple and at most four other edges separates $G$ into two subgraphs $G_1$ and $G_2$ as shown in \cref{fig:4planarUpperbound}. Based on this observation, we apply induction on the number of such triples in $G$. Note that in the base case, i.e., no triples of pairwise crossing edges exist, the graph is quasiplanar. \arxapp{Refer to \cref{app:smallK} for details.}{}
\end{proof}

\subsection{$2$-Layer $5$-Planar Graphs}

We first provide a lower bound construction for this class of graphs:

\begin{theorem}
\label{thm:5planarLowerbound}
There exist infinitely many $2$-layer $5$-planar graphs with $n$ vertices and  $\frac{9}{4}n-\frac{9}{2}$ edges.
\end{theorem}

\begin{proof}
\begin{figure}[t]
\centering
  \begin{subfigure}[b]{0.7\textwidth}
\centering \includegraphics[scale=1,page=14]{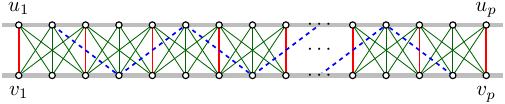}
\caption{}
\label{fig:5planarLowerbound}
\end{subfigure}
\hfil
  \begin{subfigure}[b]{0.28\textwidth}
\centering 
\includegraphics[scale=1,page=41]{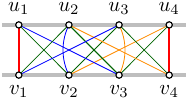}
\caption{}
\label{fig:thomasSpecialGraph}
\end{subfigure}
\caption{(a)~A family of $5$-planar graphs on $n=2p$ vertices with $\frac{9}{4}n-\frac{9}{2}$ edges. (b)~Graph $\mathcal{S}$ with $n=8$ vertices and $m=
14 > \frac{9}{4}\cdot8-\frac{9}{2}=13.5$ edges.}
\label{fig:5planarLowerbounds}
\end{figure}
We augment the construction from \cref{thm:4planarLowerbound}
by a path of length $\beta-1$, where $\beta$ is the number of $K_{3,3}$ subgraphs;
see the dashed blue edges in \cref{fig:5planarLowerbound}.
The obtained graph has $n=4\beta+2$ vertices and $m=9\beta=\frac{9}{4}n-\frac{9}{2}$ edges.
\end{proof}

For the specific value $n=8$, we can provide a denser lower bound construction.

\begin{observation}\label{obs:thomasSpecialGraph}
There exists a topological $2$-layer $5$-planar graph $\mathcal{S}$ with $n=8$ vertices and $m=
14 > \frac{9}{4}n-\frac{9}{2}$ edges; see \cref{fig:thomasSpecialGraph}.
\end{observation}

We show that the graph $\mathcal{S}$ is in fact an exception, by demonstrating that the lower bound construction in \cref{thm:5planarLowerbound} is tight for all other values of $n$.
%


\begin{restatable}{theorem}{fivePlanarUpperbound}
\label{thm:5planarUpperbound}
Any $2$-layer $5$-planar graph on $n \geq 3$ vertices has at most $\frac{9}{4}n - \frac{9}{2}$ edges, except for graph $\mathcal{S}$ which has $8$ vertices and $14$ edges.
\end{restatable}

\begin{proof}[Sketch]
\begin{figure}[t]
  \centering
  \begin{subfigure}[b]{0.35\textwidth}
\centering \includegraphics[scale=1,page=15]{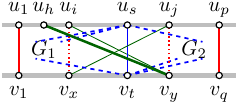}
\caption{}
\label{fig:5planarUpperbound1}
\end{subfigure}
\hfil
  \begin{subfigure}[b]{0.355\textwidth}
\centering \includegraphics[scale=1,page=16]{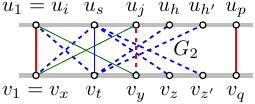}
\caption{}
\label{fig:5planarUpperbound2}
\end{subfigure}
\hfil
  \begin{subfigure}[b]{0.275\textwidth}
\centering \includegraphics[scale=1,page=42]{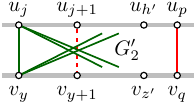}
\caption{}
\label{fig:5planarUpperbound3}
\end{subfigure}
  \caption{
    (a)~A triple $(u_i,v_y)$, $(u_s,v_t)$, $(u_j,v_x)$ of pairwise crossing edges  and at most six other edges separates an optimal $2$-layer $5$-planar graph into subgraphs $G_1$ and $G_2$.
    If $G_1$ consists of a single edge, (b)~there can be edges $(u_s,v_z)$, $(u_s,v_{z'})$, $(v_t,u_h)$, $(v_t,u_{h'})$, in which case  
    (c)~$G_2$ consists of a graph $G_2'$, vertices $u_j$, $v_y$ and at most four of the green edges.}
\label{fig:5planarUpperbound}
\end{figure}
First observe that the theorem is clearly fulfilled if $G=\mathcal{S}$. Otherwise, we apply an argument similar to the proof of \cref{thm:4planarUpperbound}. Namely, we first prove that if there is a triple of pairwise crossing edges in an optimal topological $2$-layer $5$-planar graph, the removal of few edges separates the graph into two components $G_1$ and $G_2$; see \cref{fig:5planarUpperbound1}. We then apply induction on the number of such triples in $G$. In particular, we consider some special cases, namely $G_1$ could be $\mathcal{S}$ or a single edge; see also \cref{fig:5planarUpperbound2}. In the latter case, we also investigate the structure of graph $G_2$ in more careful detail to prove our result; see also \cref{fig:5planarUpperbound3}. \arxapp{Refer to \cref{app:smallK} for more details.}{}
\end{proof}

\section{A Crossing Lemma and General Density Bounds} 
\label{sec:densityLargeK}
In this section we generalize the well-known Crossing Lemma~\cite{ajtai82,ErGu73,Le83}
to a meta Crossing Lemma for general graphs (\cref{thm:crossingLemma}), which also yields a density upper bound
for $k$-planar graphs.
We denote by $\mathcal{R}$ a restriction on graphs, e.g., $\mathcal{R}$ can be ``bipartite'' or ``$2$-layer''.
We assume that for a fixed $t>0$, there are $\alpha_i,\beta_i\in \mathbb{R}$ for $i\in \{0,\ldots,t-1\}$ such that $m\le \alpha_i n - \beta_i$ is an upper bound
for the number of edges in $\mathcal{R}$-restricted $i$-planar graphs.
Let $\alpha := \sum_{i=0}^{t-1} \alpha_i$ and $\beta := \sum_{i=0}^{t-1} \beta_i$.
The proof of the next theorem follows the probabilistic technique of 
Chazelle, Sharir and Welzl (see e.g.~\cite[Chapter 35]{AiZi98})\arxapp{; see also \cref{app:densityLargeK}}{}.

\begin{restatable}{theorem}{crossingLemma}
\label{thm:crossingLemma}
Let $G$ be a simple $\mathcal{R}$-restricted graph with $n \geq 4$ vertices and $m \geq \frac{3\alpha}{2t}n$ edges.
The following inequality holds for the crossing number $cr(G)$:
\begin{linenomath}
\begin{equation}
\label{eq:crossingLemma}
cr(G) \geq \frac{4 t^3}{27\alpha^2} \frac{m^3}{n^2}.
\end{equation}
\end{linenomath}
\end{restatable}

The meta Crossing Lemma is used to obtain the following theorem regarding the density. We follow closely the proof for corresponding
statements for $k$-planar and bipartite $k$-planar graphs~\cite{DBLP:journals/corr/Ackerman15,DBLP:conf/isaac/AngeliniB0PU18}\arxapp{ in \cref{app:densityLargeK}}{}.

\begin{restatable}{theorem}{kPlanarUpperbound}
\label{thm:kPlanarUpperbound}
Let $G$ be a simple $\mathcal{R}$-restricted $k$-planar graph with $n \geq 4$ vertices for some $k \ge t$.
Then
\begin{linenomath}
$$m \leq   \max \left\{ 1, \sqrt{\frac{3}{2t}}\sqrt{k} \right\}  \cdot \frac{3\alpha}{2t}n. $$
\end{linenomath}
\end{restatable}

We apply \cref{thm:crossingLemma,thm:kPlanarUpperbound} to 2-layer $k$-planar graphs for $t=6$.
By~\cite{DBLP:journals/ipl/Didimo13},
\cref{thm:2planarUpperbound,thm:3planarUpperbound,thm:4planarUpperbound,thm:5planarUpperbound},
we have
$(\alpha_0, \alpha_1, \alpha_2, \alpha_3, \alpha_4, \alpha_5)
=(1,        \tfrac32, \frac53,  2,        2,        \frac94)$,
yielding $\alpha = \frac{125}{12}$. 
By substituting the numbers in \cref{thm:crossingLemma} we obtain the following.

\begin{corollary}
\label{coro:crossingLemma}
Let $G$ be a simple $2$-layer graph with $n \geq 4$ vertices and $m \geq \frac{125}{48}n$ edges.
Then, the following inequality holds for the crossing number $cr(G)$:
\begin{linenomath}
\begin{equation*}
\label{eq:crossingLemmaImproved}
cr(G) \geq \frac{4.608}{15.625} \frac{m^3}{n^2} \approx 0.295 \frac{m^3}{n^2}.
\end{equation*}
\end{linenomath}
\end{corollary}

By plugging the result into \cref{thm:kPlanarUpperbound} we obtain.

\begin{corollary}
\label{coro:kPlanarUpperbound}
Let $G$ be a simple $2$-layer $k$-planar graph with $n \geq 4$ vertices for some $k > 5$. Then
\begin{linenomath}
$$ m \leq \max \left\{ \frac{125}{48}, \frac{125}{96} \sqrt{k} \right\} \cdot n. $$
\end{linenomath}
\end{corollary}

Note that for $2$-layer $6$-planar graphs, \cref{coro:kPlanarUpperbound} certifies that $m \leq 3.19n$. We can show that there is only a gap of $0.69n$ towards an optimal solution:

\begin{theorem}
\label{thm:6planarLowerbound}
There exist infinitely many $2$-layer $6$-planar graphs with $n$ vertices and $\frac{5}{2}n-6$ edges.
\end{theorem}

\begin{proof}
\begin{figure}[t]
\centering
\includegraphics[scale=1,page=17]{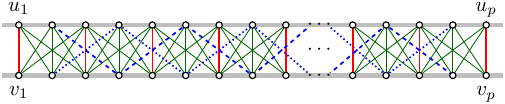}
\caption{A family of $6$-planar graphs on $n=2p$ vertices with $\frac{5}{2}n-6$ edges.}
\label{fig:6planarLowerbound}
\end{figure}
We augment the construction from \cref{thm:5planarLowerbound} by a path of length $\beta-1$,
where $\beta$ is the number of $K_{3,3}$ subgraphs; refer to the dotted blue path
in \cref{fig:6planarLowerbound}.
The obtained graph has $n=4\beta+2$ vertices and $m=10\beta-1=\frac{5}{2}n-6$ edges.
\end{proof}

In the next theorem, we additionally show that the multiplicative constant from \cref{coro:kPlanarUpperbound} is within a factor of $1.84$ of the optimal achievable upper bound.

\begin{restatable}{theorem}{kPlanarLowerbound}
\label{thm:kPlanarLowerbound}
For any $k$, there exist infinitely many $2$-layer $k$-planar graphs with $n$ vertices and
$m=\left\lfloor\sqrt{k/2}\right\rfloor n - \mathcal{O}(f(k)) \approx 0.707 \sqrt{k} n - \mathcal{O}(f(k))$ edges.
\end{restatable}

\begin{proof}[Sketch]
\begin{figure}[t]
  \centering
    \begin{subfigure}[b]{0.35\textwidth}
      \centering \includegraphics[scale=1,page=37]{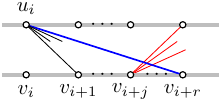}
      \caption{}
      \label{fig:generalLowerBound1}
    \end{subfigure}
    \hfil
    \begin{subfigure}[b]{0.4\textwidth}
      \centering \includegraphics[scale=1,page=38]{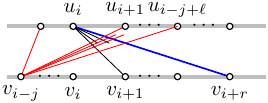}
      \caption{}
      \label{fig:generalLowerBound2}
    \end{subfigure}
    
		\begin{subfigure}[b]{0.4\textwidth}
		  \centering \includegraphics[scale=1,page=39]{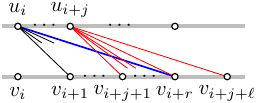}
			\caption{}
			\label{fig:generalLowerBound3}
		\end{subfigure}
		\begin{subfigure}[b]{0.55\textwidth}
		  \centering \includegraphics[scale=1,page=40]{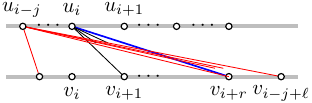}
			\caption{}
			\label{fig:generalLowerBound4}
		\end{subfigure}
	\caption{Illustrations for the proof of \cref{thm:kPlanarLowerbound}.}
  \label{fig:generalLowerBound}
\end{figure}
We choose $p=q$ and a parameter $\ell = \lfloor \sqrt{k/2} \rfloor$. We connect vertex $u_i$ to the $\ell$ vertices $v_{i+1}\ldots,v_{i+\ell}$ and vertex $v_i$ to vertices $u_{i+1}\ldots,u_{i+\ell}$. Note that by symmetry, $u_i$ is also incident to the $\ell$ vertices $v_{i-1}\ldots,v_{i-\ell}$ and vertex $v_i$ to vertices $u_{i-1}\ldots,u_{i-\ell}$.  Clearly, this gives the density bound in the statement of the theorem. Then, we consider an edge $(u_i,v_{i+r})$ and the crossings it forms with edges incident to some other vertices; see \cref{fig:generalLowerBound}. This allows us to establish that each edge has at most $k$ crossings. \arxapp{For details, see \cref{app:densityLargeK}.}{}
\end{proof}

\section{Properties of $2$-Layer $k$-Planar Graphs}
\label{sec:prop}
In this section, we present some properties of $2$-layer $k$-planar graphs.

In \cref{thm:3planarUpperbound}, we have established that every optimal $2$-layer $3$-planar graph is ($3$-)quasiplanar, which is also the case in the general, non-layered, drawing model~\cite{DBLP:conf/compgeom/Bekos0R17}. 
A more general relationship between the classes of $k$-planar and $h$-quasiplanar graphs was uncovered in~\cite{DBLP:journals/jctb/AngeliniBBLBDHL20}, where it is proven that every $k$-planar graph is $(k+1)$-quasiplanar, for every $k \geq 2$. Next, we show that for 2-layer drawings an even stronger relationship holds.

\begin{theorem}\label{th:relationship}
For $k\ge 3$, every $2$-layer $k$-planar graph is $2$-layer $\left\lceil\frac{2}{3}k+2\right\rceil$-quasiplanar.
Further, every $2$-layer $2$-planar graph is $2$-layer (3-)quasiplanar.
\end{theorem}

\begin{figure}[t]
  \centering
  \includegraphics[scale=1,page=10]{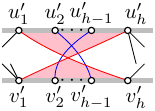}
  \caption{A set of $h$ pairwise crossing edges in a topological $2$-layer graph.}
  \label{fig:3planarNonquasi}
\end{figure}
\begin{proof}
Let $G$ be a topological $2$-layer $k$-planar graph, with $k\geq3$, which we assume w.l.o.g. to be connected.
Suppose for a contradiction that $G$ contains $h:=\lceil\frac{2}{3}k+2\rceil$ mutually crossing edges
$(u_i',v_{h+1-i}')$ for $1 \leq i \leq h$ in $G$, such that $u_1',\ldots,u_{h}'$ and $v_1',\ldots,v_{h}'$
appear in this order in $u_1,\ldots,u_p$ and $v_1,\ldots,v_q$, respectively.
Observe that $(u_1',v_{h}')$ and $(v_1',u_{h}')$ have $h-1$ crossings from this $h$-tuple.
Moreover, both endvertices of all the $h-2$ edges $(u_i',v_{h+1-i}')$, for $i=2,\dots,h-1$,
are located in regions bounded by $e^{(1)}:=(u_1',v_{h}')$ and $e^{(2)}:=(v_1',u_{h}')$; see \cref{fig:3planarNonquasi}.
Since $G$ is connected, for each $2\le i\le h-1$, the edge $(u_i',v_{h+1-i}')$ is adjacent to another edge $e_i$.
Note that either $e_i=e_j$ for some $j\ne i$, and $e_i$ crosses $e^{(1)}$ and $e^{(2)}$,
or $e_i\ne e_j$ for all $j\ne i$, and $e_i$ crosses one of $e^{(1)}$ and $e^{(2)}$.
This implies $h-2$ additional crossings for $\{e^{(1)},e^{(2)}\}$, and, consequently,
$e^{(1)}$ or $e^{(2)}$ is crossed by at least $h-1+\lceil(h-2)/2\rceil$ edges.
We obtain
\begin{linenomath}
$ h-1+\lceil(h-2)/2\rceil \ge \tfrac32h-2 \ge \tfrac32\left(\tfrac23k+2\right)-2 = k+1 $
\end{linenomath}
crossings for $e^{(1)}$ or $e^{(2)}$, a contradiction.

For the case $k=2$, assume that $G$ contains three mutually crossing edges
$e_1=(u_1',v_3')$, $e_2=(u_2',v_2')$ and $e_3=(u_3',v_1')$, such that $u_1', u_2', u_3'$ and
$v_1', v_2', v_3'$. appear in this order in $u_1,\ldots,u_p$ and $v_1,\ldots,v_q$, respectively. As $e_1$ and $e_3$ are already crossed twice, $e_2$ represents a connected 
component; contradiction.
\end{proof}
%

Next, we show that the pathwidth of $2$-layer $k$-planar graphs is bounded by $k+1$. We point out that similar results are known for layered graphs  with a bounded total number of crossings~\cite{Dujmovic2008} and for layered fan-planar graphs~\cite{DBLP:journals/corr/abs-2002-09597}, and that these bounds do not have any implication on $2$-layer $k$-planar graphs. 

\begin{theorem}
\label{thm:pathwidth}
Every $2$-layer $k$-planar graph has pathwidth at most $k+1$.
\end{theorem}

\begin{proof}
Let $G$ be a topological $2$-layer $k$-planar graph with parts $U$ and $V$. We first define a total ordering $\prec$ on the edges as follows: We say that edge $e_1=(u_i,v_x)$ precedes edge $e_2=(u_j,v_y)$, or $e_1 \prec e_2$,  if $u_i,u_j \in U$ and either 
\begin{inparaenum}[(i)]
\item $i < j$, or
\item $i=j$ and $x<y$.
\end{inparaenum}
Let $E=(e_1,\ldots,e_m)$ be the set of edges ordered with respect to $\prec$. 
Let $e_i=(u_s,v_t)$ be an edge and let $v_y$ be a vertex in $V$. Further let $e_{y^-}$ and $e_{y^+}$ be the first and the last edge incident to $v_y$ in $\prec$, respectively. We call $v_y$ \emph{related} to $e_i$ if $v_y$ is incident to an edge crossing $e_i$ and if $y^- < i < y^+$. 
For every edge $e_i=(u_s,v_t) \in E$, we construct a bag $B_i$ that contains $u_s$, $v_t$ and all the (at most $k$) related vertices of $e_i$. Then, we connect $B_i$ to bags $B_{i-1}$ and $B_{i+1}$ (if they exist), obtaining a path of bags $P$.

In the following we show that $P$ is a valid path decomposition of $G$. Since we assigned at most $k+2$ vertices
to each bag of $P$ the width of $P$ is at most $k+1$.
Properties P.\ref{prop:treedecomp1} and P.\ref{prop:treedecomp3} of a tree decomposition are fulfilled
for $P$ by construction. We may assume that $G$ is connected, otherwise we compute a path decomposition
for each connected component and link the obtained vertex disjoint paths.
Hence also P.\ref{prop:treedecomp2} is fulfilled. Moreover, by the choice of $\prec$,
all the edges incident to a vertex $u_i \in U$ occur in a consecutive sequence,
i.e. $u_i$ is incident to edges $e_j,\ldots,e_k$ for some $1 \leq j \leq k \leq m$
and then $u_i$ appears in all of bags $B_j,\ldots,B_k$,
which is a subpath of $P$. Therefore, Property P.\ref{prop:treedecomp4} also holds for all vertices in $U$.


\begin{figure}[t]
\centering
  \begin{subfigure}[b]{0.23\textwidth}
\centering \includegraphics[scale=1,page=43]{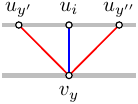}
\caption{}
\label{fig:pathwidth1}
\end{subfigure}
\hfil
  \begin{subfigure}[b]{0.23\textwidth}
\centering \includegraphics[scale=1,page=44]{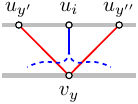}
\caption{}
\label{fig:pathwidth2}
\end{subfigure}
\hfil
  \begin{subfigure}[b]{0.23\textwidth}
\centering \includegraphics[scale=1,page=45]{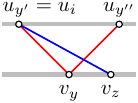}
\caption{}
\label{fig:pathwidth3}
\end{subfigure}
\hfil
  \begin{subfigure}[b]{0.23\textwidth}
\centering \includegraphics[scale=1,page=46]{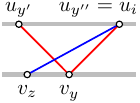}
\caption{}
\label{fig:pathwidth4}
\end{subfigure}
\caption{Illustrations for the proof of \cref{thm:pathwidth}.}
\label{fig:pathwidth}
\end{figure}

It remains to show that Property P.\ref{prop:treedecomp4} holds for every vertex $v_y \in V$. Let $e_{y^-}=(u_{y'},v_y)$ and $e_{y^+}=(u_{y''},v_y)$.
Note that each of the edges
$e_{y^-},e_{y^-+1},\ldots,e_{y^+}$ is either incident to $v_y$ (see \cref{fig:pathwidth1}),
or it crosses one of $e_{y^-}$ and $e_{y^+}$, since its endvertex in $U$ is some $u_i$ with $y' \leq i \leq y''$; see \cref{fig:pathwidth2,fig:pathwidth3,fig:pathwidth4}. Note that for the endvertex $v_z$ in $V$ necessarily $z>y$  if $u_i = u_{y'}$  or $z<y$ if $u_i = u_{y''}$ by definition of $\prec$; see \cref{fig:pathwidth3} or \cref{fig:pathwidth4}, respectively. Hence $v_y$ belongs to all bags $B_{y^-},B_{y^-+1},\ldots,B_{y^+}$ and P.\ref{prop:treedecomp4} holds. The statement follows.
\end{proof}

\section{Conclusions}

We gave results for $2$-layer $k$-planar graphs regarding their density, relationship to $2$-layer $h$-quasiplanar graphs, and pathwidth. Tight density bounds for $2$-layer $k$-planar graphs with $k = 6$ may be achievable following similar arguments to the proof of \cref{thm:5planarUpperbound}, which would also improve upon our results for the Crossing Lemma, and in turn on the density for general~$k$. Moreover, a better lower bound for general $k$ may exist. The relationship to other beyond-planar graph classes is also of interest. With respect to the pathwidth, we conjecture that our upper bound is tight. Finally, the recognition and characterization of $2$-layer $k$-planar graphs remain important open problems.

\bibliographystyle{splncs04}
\bibliography{twoLayer}

\begin{thebibliography}{10}
\providecommand{\url}[1]{\texttt{#1}}
\providecommand{\urlprefix}{URL }
\providecommand{\doi}[1]{https://doi.org/#1}

\bibitem{Ackerman09}
Ackerman, E.: On the maximum number of edges in topological graphs with no four
  pairwise crossing edges. Discret. Comput. Geom.  \textbf{41}(3),  365--375
  (2009), \url{https://doi.org/10.1007/s00454-009-9143-9}

\bibitem{DBLP:journals/corr/Ackerman15}
Ackerman, E.: On topological graphs with at most four crossings per edge.
  Comput. Geom.  \textbf{85} (2019),
  \url{https://doi.org/10.1016/j.comgeo.2019.101574}

\bibitem{AckermanT07}
Ackerman, E., Tardos, G.: On the maximum number of edges in quasi-planar
  graphs. J. Comb. Theory, Ser. {A}  \textbf{114}(3),  563--571 (2007),
  \url{https://doi.org/10.1016/j.jcta.2006.08.002}

\bibitem{AgarwalAPPS97}
Agarwal, P.K., Aronov, B., Pach, J., Pollack, R., Sharir, M.: Quasi-planar
  graphs have a linear number of edges. Combinatorica  \textbf{17}(1), ~1--9
  (1997). \doi{10.1007/BF01196127}

\bibitem{AiZi98}
Aigner, M., Ziegler, G.M.: Proofs from {THE} {BOOK} {(3rd.} ed.). Springer
  (2004)

\bibitem{ajtai82}
Ajtai, M., Chvátal, V., Newborn, M., Szemerédi, E.: Crossing-free subgraphs.
  In: Hammer, P.L., Rosa, A., Sabidussi, G., Turgeon, J. (eds.) Theory and
  Practice of Combinatorics, North-Holland Mathematics Studies, vol.~60, pp. 9
  -- 12. North-Holland (1982),
  \url{http://www.sciencedirect.com/science/article/pii/S0304020808734844}

\bibitem{DBLP:journals/jctb/AngeliniBBLBDHL20}
Angelini, P., Bekos, M.A., Brandenburg, F.J., {Da Lozzo}, G., {Di Battista},
  G., Didimo, W., Hoffmann, M., Liotta, G., Montecchiani, F., Rutter, I.,
  T{\'{o}}th, C.D.: Simple \emph{k}-planar graphs are simple
  (\emph{k}+1)-quasiplanar. J. Comb. Theory, Ser. {B}  \textbf{142},  1--35
  (2020), \url{https://doi.org/10.1016/j.jctb.2019.08.006}

\bibitem{DBLP:journals/tcs/AngeliniBFK20}
Angelini, P., Bekos, M.A., F{\"{o}}rster, H., Kaufmann, M.: On {RAC} drawings
  of graphs with one bend per edge. Theor. Comput. Sci.  \textbf{828-829},
  42--54 (2020), \url{https://doi.org/10.1016/j.tcs.2020.04.018}

\bibitem{DBLP:conf/isaac/AngeliniB0PU18}
Angelini, P., Bekos, M.A., Kaufmann, M., Pfister, M., Ueckerdt, T.:
  Beyond-planarity: Tur{\'{a}}n-type results for non-planar bipartite graphs.
  In: {ISAAC}. LIPIcs, vol.~123, pp. 28:1--28:13. Schloss Dagstuhl (2018).
  \doi{10.4230/LIPIcs.ISAAC.2018.28}

\bibitem{AuerBBGHNR16}
Auer, C., Bachmaier, C., Brandenburg, F.J., Glei{\ss}ner, A., Hanauer, K.,
  Neuwirth, D., Reislhuber, J.: Outer 1-planar graphs. Algorithmica
  \textbf{74}(4),  1293--1320 (2016). \doi{10.1007/s00453-015-0002-1}

\bibitem{avital-66}
Avital, S., Hanani, H.: Graphs. Gilyonot Lematematika  \textbf{3}, ~2--8 (1966)

\bibitem{BekosCGHK14}
Bekos, M.A., Cornelsen, S., Grilli, L., Hong, S., Kaufmann, M.: On the
  recognition of fan-planar and maximal outer-fan-planar graphs. Algorithmica
  \textbf{79}(2),  401--427 (2017),
  \url{https://doi.org/10.1007/s00453-016-0200-5}

\bibitem{DBLP:conf/gd/Bekos0R16}
Bekos, M.A., Kaufmann, M., Raftopoulou, C.N.: On the density of non-simple
  3-planar graphs. In: Graph Drawing. LNCS, vol.~9801, pp. 344--356. Springer
  (2016), \url{https://doi.org/10.1007/978-3-319-50106-2_27}

\bibitem{DBLP:conf/compgeom/Bekos0R17}
Bekos, M.A., Kaufmann, M., Raftopoulou, C.N.: On optimal 2- and 3-planar
  graphs. In: Aronov, B., Katz, M.J. (eds.) 33rd International Symposium on
  Computational Geometry, SoCG 2017, July 4-7, 2017, Brisbane, Australia.
  LIPIcs, vol.~77, pp. 16:1--16:16. Schloss Dagstuhl - Leibniz-Zentrum
  f{\"{u}}r Informatik (2017),
  \url{https://doi.org/10.4230/LIPIcs.SoCG.2017.16}

\bibitem{DBLP:journals/corr/abs-2002-09597}
Biedl, T.C., Chaplick, S., Kaufmann, M., Montecchiani, F., N{\"{o}}llenburg,
  M., Raftopoulou, C.N.: Layered fan-planar graph drawings. In: Esparza, J.,
  Kr{\'{a}}l', D. (eds.) 45th International Symposium on Mathematical
  Foundations of Computer Science, {MFCS} 2020, August 24-28, 2020, Prague,
  Czech Republic. LIPIcs, vol.~170, pp. 14:1--14:13. Schloss Dagstuhl -
  Leibniz-Zentrum f{\"{u}}r Informatik (2020).
  \doi{10.4230/LIPIcs.MFCS.2020.14},
  \url{https://doi.org/10.4230/LIPIcs.MFCS.2020.14}

\bibitem{DBLP:journals/jgaa/BinucciCDGKKMT17}
Binucci, C., Chimani, M., Didimo, W., Gronemann, M., Klein, K.,
  Kratochv{\'{\i}}l, J., Montecchiani, F., Tollis, I.G.: Algorithms and
  characterizations for 2-layer fan-planarity: From caterpillar to stegosaurus.
  J. Graph Algorithms Appl.  \textbf{21}(1),  81--102 (2017),
  \url{https://doi.org/10.7155/jgaa.00398}

\bibitem{BinucciGDMPST15}
Binucci, C., {Di Giacomo}, E., Didimo, W., Montecchiani, F., Patrignani, M.,
  Symvonis, A., Tollis, I.G.: Fan-planarity: Properties and complexity. Theor.
  Comput. Sci.  \textbf{589},  76--86 (2015). \doi{10.1016/j.tcs.2015.04.020}

\bibitem{DBLP:journals/jct/CapoyleasP92}
Capoyleas, V., Pach, J.: A {T}ur{\'{a}}n-type theorem on chords of a convex
  polygon. J. Comb. Theory, Ser. {B}  \textbf{56}(1),  9--15 (1992)

\bibitem{DBLP:conf/gd/ChaplickKLLW17}
Chaplick, S., Kryven, M., Liotta, G., L{\"{o}}ffler, A., Wolff, A.: Beyond
  outerplanarity. In: Frati, F., Ma, K. (eds.) Graph Drawing and Network
  Visualization - 25th International Symposium, {GD} 2017, Boston, MA, USA,
  September 25-27, 2017, Revised Selected Papers. Lecture Notes in Computer
  Science, vol. 10692, pp. 546--559. Springer (2017),
  \url{https://doi.org/10.1007/978-3-319-73915-1\_42}

\bibitem{DehkordiEHN16}
Dehkordi, H.R., Eades, P., Hong, S., Nguyen, Q.H.: Circular right-angle
  crossing drawings in linear time. Theor. Comput. Sci.  \textbf{639},  26--41
  (2016). \doi{10.1016/j.tcs.2016.05.017}

\bibitem{GiacomoDEL14}
{Di Giacomo}, E., Didimo, W., Eades, P., Liotta, G.: 2-layer right angle
  crossing drawings. Algorithmica  \textbf{68}(4),  954--997 (2014).
  \doi{10.1007/s00453-012-9706-7}

\bibitem{DBLP:journals/ipl/Didimo13}
Didimo, W.: Density of straight-line 1-planar graph drawings. Inf. Process.
  Lett.  \textbf{113}(7),  236--240 (2013),
  \url{https://doi.org/10.1016/j.ipl.2013.01.013}

\bibitem{DidimoEL11}
Didimo, W., Eades, P., Liotta, G.: Drawing graphs with right angle crossings.
  Theor. Comput. Sci.  \textbf{412}(39),  5156--5166 (2011).
  \doi{10.1016/j.tcs.2011.05.025}

\bibitem{Didimo2013}
Didimo, W., Liotta, G.: The crossing-angle resolution in graph drawing. In:
  Thirty Essays on Geometric Graph Theory, pp. 167--184. Springer (2013).
  \doi{10.1007/978-1-4614-0110-0\_10}

\bibitem{DBLP:journals/csur/DidimoLM19}
Didimo, W., Liotta, G., Montecchiani, F.: A survey on graph drawing beyond
  planarity. {ACM} Comput. Surv.  \textbf{52}(1),  4:1--4:37 (2019),
  \url{https://doi.org/10.1145/3301281}

\bibitem{Dujmovic2008}
Dujmovi{\'{c}}, V., Fellows, M.R., Kitching, M., Liotta, G., McCartin, C.,
  Nishimura, N., Ragde, P., Rosamond, F., Whitesides, S., Wood, D.R.: On the
  parameterized complexity of layered graph drawing. Algorithmica
  \textbf{52}(2),  267--292 (Oct 2008),
  \url{https://doi.org/10.1007/s00453-007-9151-1}

\bibitem{EadesL13}
Eades, P., Liotta, G.: Right angle crossing graphs and 1-planarity. Discret.
  Appl. Math.  \textbf{161}(7-8),  961--969 (2013),
  \url{https://doi.org/10.1016/j.dam.2012.11.019}

\bibitem{ErGu73}
Erd{\H{o}}s, P., Guy, R.: Crossing number problems. The American Mathematical
  Monthly  \textbf{80},  52--58 (01 1973). \doi{10.2307/2319261}

\bibitem{DBLP:conf/compgeom/FoxP08}
Fox, J., Pach, J.: Coloring {$K_k$}-free intersection graphs of geometric
  objects in the plane. In: Symposium on Computational Geometry. pp. 346--354.
  {ACM} (2008)

\bibitem{DBLP:journals/siamdm/FoxPS13}
Fox, J., Pach, J., Suk, A.: The number of edges in k-quasi-planar graphs.
  {SIAM} J. Discrete Math.  \textbf{27}(1),  550--561 (2013)

\bibitem{HongEKLSS15}
Hong, S., Eades, P., Katoh, N., Liotta, G., Schweitzer, P., Suzuki, Y.: A
  linear-time algorithm for testing outer-1-planarity. Algorithmica
  \textbf{72}(4),  1033--1054 (2015). \doi{10.1007/s00453-014-9890-8}

\bibitem{HongN15}
Hong, S., Nagamochi, H.: A linear-time algorithm for testing full
  outer-2-planarity. Discret. Appl. Math.  \textbf{255},  234--257 (2019).
  \doi{10.1016/j.dam.2018.08.018}

\bibitem{KaufmannU14}
Kaufmann, M., Ueckerdt, T.: The density of fan-planar graphs. CoRR
  \textbf{1403.6184} (2014)

\bibitem{DBLP:journals/jgt/KorzhikM13}
Korzhik, V.P., Mohar, B.: Minimal obstructions for 1-immersions and hardness of
  1-planarity testing. Journal of Graph Theory  \textbf{72}(1),  30--71 (2013),
  \url{https://doi.org/10.1002/jgt.21630}

\bibitem{Le83}
Leighton, F.T.: Complexity Issues in VLSI: Optimal Layouts for the
  Shuffle-exchange Graph and Other Networks. MIT Press, Cambridge, MA, USA
  (1983)

\bibitem{PachRTT06}
Pach, J., Radoicic, R., Tardos, G., T{\'{o}}th, G.: Improving the crossing
  lemma by finding more crossings in sparse graphs. Discret. Comput. Geom.
  \textbf{36}(4),  527--552 (2006),
  \url{https://doi.org/10.1007/s00454-006-1264-9}

\bibitem{DBLP:conf/jcdcg/PachRT02}
Pach, J., Radoicic, R., T{\'o}th, G.: Relaxing planarity for topological
  graphs. In: {JCDCG}. vol.~2866, pp. 221--232 (2002).
  \doi{10.1007/978-3-540-44400-8\_24}

\bibitem{DBLP:journals/algorithmica/PachSS96}
Pach, J., Shahrokhi, F., Szegedy, M.: Applications of the crossing number.
  Algorithmica  \textbf{16}(1),  111--117 (1996)

\bibitem{PachT97}
Pach, J., T{\'o}th, G.: Graphs drawn with few crossings per edge. Combinatorica
   \textbf{17}(3),  427--439 (1997). \doi{10.1007/BF01215922}

\bibitem{MR0187232}
Ringel, G.: Ein {S}echsfarbenproblem auf der {K}ugel. Abh. Math. Sem. Univ.
  Hamb.  \textbf{29},  107--117 (1965). \doi{10.1007/BF02996313}

\bibitem{DBLP:series/sseke/Sugiyama02}
Sugiyama, K.: Graph Drawing and Applications for Software and Knowledge
  Engineers, Series on Software Engineering and Knowledge Engineering, vol.~11.
  WorldScientific (2002), \url{https://doi.org/10.1142/4902}

\bibitem{DBLP:journals/tsmc/SugiyamaTT81}
Sugiyama, K., Tagawa, S., Toda, M.: Methods for visual understanding of
  hierarchical system structures. IEEE Transactions on Systems, Man, and
  Cybernetics, SMC-11  \textbf{11}(2),  109--125 (1981).
  \doi{10.1109/TSMC.1981.4308636}

\bibitem{DBLP:journals/comgeo/SukW15}
Suk, A., Walczak, B.: New bounds on the maximum number of edges in
  $k$-quasi-planar graphs. Comput. Geom.  \textbf{50},  24--33 (2015)

\bibitem{DBLP:journals/dcg/Valtr98}
Valtr, P.: On geometric graphs with no $k$ pairwise parallel edges. Discrete
  Comput. Geom.  \textbf{19}(3),  461--469 (1998)

\bibitem{gdInvitedTalk}
Walczak, B.: Old and new challenges in coloring graphs with geometric
  representations. In: Archambault, D., T{\'{o}}th, C.D. (eds.) Graph Drawing
  and Network Visualization - 27th International Symposium, {GD} 2019, Prague,
  Czech Republic, September 17-20, 2019, Proceedings. Lecture Notes in Computer
  Science, vol. 11904. Springer (2019), {I}nvited talk.

\end{thebibliography}
\newpage 
\arxapp{
\appendix

\section{Omitted Proofs from \cref{sec:smallK}}
\label{app:smallK}

We start with the full proof of \cref{lem:rational-density}.

\rationalDensity*

\begin{proof}
First observe that the number of edges in an $n$-vertex $k$-planar graph can be upper bounded by a linear function $g(n)$~\cite{DBLP:journals/corr/Ackerman15,PachT97}. Thus in an $n$-vertex $2$-layer $k$-planar graph, there can be at most $f(n) \leq g(n)$ edges for a linear function $f(n)$. If the multiplicative factor $a_k$ in $f(n)=a_kn-o(n)$ is irrational, one can choose a slightly larger rational multiplicative factor $a_k'$. So assume w.l.o.g. that $a_k$ is rational. Since optimal $2$-layer $1$-planar graphs can have $\frac{3}{2}n-2$ edges~\cite{GiacomoDEL14}, it follows that $a_k \geq \frac{3}{2}$.
	
Because $a_k$ is rational and the number of edges in a graph must be integer, there is a limited number of possible rational differences between the maximum density and the function $a_kn$. We choose the smallest divergence as our additive constant $b_k$.
\end{proof}

The following lemma is a key ingredient in the proof of the density upper bound for $2$-layer $3$-planar graphs in \cref{thm:3planarUpperbound}.

\begin{lemma}
\label{lem:3planarNonquasi}
Let $(u_i,v_y)$, $(u_s,v_t)$ and $(u_j,v_x)$ be a triple of pairwise crossing edges
in a topological $2$-layer $3$-planar graph such that $1 \leq i < s < j \leq p$
and $1 \leq x < t < y \leq q$. Then the number of edges incident to $u_s$ or $v_t$ is at most $3$.
\end{lemma}

\begin{proof}
Consider such a triple of edges. If $u_s$ is connected to another vertex $v \neq v_t$, edge $(u_s,v)$ crosses one of $(u_i,v_y)$ and $(u_j,v_x)$. The same is true, if $v_t$ is connected to another vertex $u \neq u_s$. Since $(u_i,v_y)$ and $(u_j,v_x)$ both have two crossings from the triple of crossing edges, $u_s$ and $v_t$ can only be incident to a total of two edges which are not $(u_s,v_t)$. 
\end{proof}

\threePlanarUpperbound*

\begin{proof}
Let $G$ be an optimal $2$-layer $3$-planar graph on $n$ vertices. By \cref{lem:3planarNonquasi}, for every triple of pairwise crossing edges there exists two vertices $u_s$ and $v_t$ which are incident to a total of at most $4$ edges. Removing $u_s$ and $v_t$ reduces the number of edges by at most $3$ and the number of vertices by $2$. Since by \cref{thm:3planarLowerbound} optimal $2$-layer $3$-planar graphs have density at least $2n-4$, the removal of $u_s$ and $v_t$ yields a denser subgraph. We conclude that $G$ contains no triple of pairwise crossing edges. Thus, $G$ is quasiplanar and has at most $2n-4$ edges for $n \geq 2$, by \cref{thm:quasiplanarDensity}.
\end{proof}

\fourPlanarUpperbound*

\begin{proof}
Consider an optimal $2$-layer $4$-planar graph $G$ with exactly $a_4n-b_4$ edges. By 
\cref{thm:quasiplanarDensity,thm:4planarLowerbound},
this graph cannot be quasiplanar and hence contains a triple of pairwise crossing edges $(u_i,v_y)$, $(u_s,v_t)$
and $(u_j,v_x)$ for some $1 \leq i < s < j \leq p$ and $1 \leq x < t < y \leq q$.
We first show  that there is no vertex $u_{s'}$ such that $i < s' < j$ and $s' \neq s$.
Assume for a contradiction that such a vertex $u_{s'}$ exists.
Each of the edges $(u_i,v_y)$ and $(u_j,v_x)$ have two crossings from the triple of pairwise crossing edges,
so each of them can only be crossed by two more edges.
Hence, there are at most five edges incident to the vertices $u_s,u_{s'},v_t$ (including the edge $(u_s,v_t)$).
Then the graph $G'$ obtained by removing vertices $u_s,u_{s'},v_t$ has $n'=n-3$ vertices and $m' \geq m -  5= a_4(n-3)-b_4+(3a_4-5) > a_4n'-b_4$ edges; a contradiction to the optimality of $G$.
Symmetrically, there is no vertex $v_{t'}$ such that $x < t' < y$ and $t' \neq t$.
So we have $s=i+1$, $j=i+2$, $t=x+1$ and $y=x+2$.

Next, consider the subgraph $G_1$ induced by $u_1,\ldots,u_i$ and $v_1,\ldots v_x$
and the subgraph $G_2$ induced by $u_{i+2},\ldots,u_p$ and $v_{x+2},\ldots v_q$.
We show that $G_1$ is connected to $G_2$ only by the edges $(u_i,v_{x+2})$, $(u_{i+2},v_x)$
and paths traversing $u_s$ or $v_t$; see \cref{fig:4planarUpperbound}.
Assume for a contradiction that there is an edge $(u_h,v_z)$ such that w.l.o.g. $1 \leq h \leq i$ and $x+2 \leq z \leq q$.
This edge would  cross $(u_s,v_t)$ and at least one of the edges $(u_i,v_{x+2})$ and $(u_{i+2},v_x)$,
say $(u_{i+2},v_x)$. Then, $(u_h,v_z)$, $(u_{i+2},v_x)$,  and $(u_s,v_t)$ would form a triple of pairwise
crossing edges where $u_i$ and $u_{i+1}$ are between $u_h$ and $u_{i+2}$;
a contradiction to the previously established claim.

We show by induction on the number of triples of pairwise crossing edges that the number of edges of $G$
is at most $2n-3$. For the base case, assume that there is no triple of pairwise crossing edges.
Then, $G$ is quasiplanar and has at most $2n-4$ edges; however, in the degenerate case of $n=2$,
it has at most $1=2n-3$ edges.

For the induction step, assume that $G$ has a triple of pairwise crossing edges that connects
subgraphs $G_1$ and $G_2$ as described above. Clearly, $G_1$ and $G_2$ have less triples of pairwise crossing edges than $G$. As mentioned before, since $(u_i,v_{x+2})$ and $(u_{i+2},v_x)$ have
two crossings from the triple of pairwise crossing edges, $u_{i+1}$ and $v_{x+1}$  can
only be incident to a total of five edges including edge $(u_{i+1},v_{x+1})$. Hence, $G$ can be split into
(possibly optimal) subgraphs $G_1$ and $G_2$ and two isolated vertices $u_{i+1}$ and $v_{x+1}$ by
removing seven edges. Let $n_1$ and $n_2$ be the number of vertices of $G_1$ and $G_2$, respectively.
Clearly, $n=n_1+n_2+2$. By induction, $G_1$ and $G_2$ have at most $2n_1-3$ and $2n_2-3$ edges, respectively.
Then,  $m \leq 2 (n_1+n_2)-6+7 = 2(n_1+n_2)+1 = 2(n-2)+1 = 2n -3$. 
\end{proof}

\fivePlanarUpperbound*

\begin{proof}
Consider an optimal topological $2$-layer $5$-planar graph $G$. If $G=\mathcal{S}$, the theorem trivially holds. Hence, assume that $G \neq \mathcal{S}$ has exactly $a_5n-b_5$ edges. By \cref{thm:quasiplanarDensity,thm:5planarLowerbound},
it cannot be quasiplanar and hence contains a triple of pairwise crossing edges $(u_i,v_y)$, $(u_s,v_t)$
and $(u_j,v_x)$ for some $1 \leq i < s < j \leq p$ and $1 \leq x < t < y \leq q$.
We first show  that there is at most one vertex $u_{s'}$ such that $i < s' < j$ and $s' \neq s$,
or at most one vertex $v_{t'}$ such that $x < t' < y$ and $t' \neq t$.
Assume for a contradiction that two such vertices $w$ and $w'$ exist. Since both edges $(u_i,v_y)$ and $(u_j,v_x)$
have two crossings from the triple of pairwise crossing edges, each of those edges can only be crossed
by three more edges. Hence, there are at most seven edges incident to vertices $u_s,v_t,w,w'$
(including the edge $(u_s,v_t)$). By removing the four vertices $u_s,v_{t},w,w'$ together
with the at most seven incident edges, we would obtain a graph $G'$ with $n'=n-4$ vertices and $m'=m-7=a_5n-b_5-7 = a_5(n-4) - b_5 + (4a_5-7)$ edges. Given $a_5 \geq \frac{9}{4}$ by \cref{thm:5planarLowerbound}, it holds that $4a_5-7 \geq 2$ and thus $m' > a_5n'-b_5$; a contradiction.

Next, consider the subgraph $G_1$ induced by vertices $u_1,\ldots,u_i$ and $v_1,\ldots v_x$
and the subgraph $G_2$ induced by vertices $u_{j},\ldots,u_p$ and $v_{y},\ldots v_q$.
We show that $G_1$ is connected to $G_2$ only by the edges $(u_i,v_{y})$, $(u_{j},v_x)$, some
paths traversing $u_s$ or $v_t$ and potentially an edge $(u_h,v_z)$ such that $h \in \{i,j\}$
or $z \in \{x,y\}$; see \cref{fig:5planarUpperbound1}.
Assume for a contradiction that there is an edge $(u_h,v_z)$ such that w.l.o.g. $1 \leq h < i$ and $y < z \leq q$.
This edge would cross all three edges $(u_i,v_{y})$, $(u_{j},v_x)$, and $(u_s,v_t)$.
Then, $(u_h,v_z)$, $(u_{j},v_x)$, and $(u_s,v_t)$ form a triple of pairwise crossing edges
where $u_i$ and $v_y$ are between $u_h$ and $u_{j}$, and between $v_x$ and $v_z$, respectively;
a contradiction to the previously established claim. By a similar argument, there can be only one such edge $(u_h,v_z)$. In the following, we assume w.l.o.g. that $h < i$ and $y = z$.

If both, $(u_i,v_y)$, $(u_j,v_x)$, and $(u_s,v_t)$ as well as $(u_h,v_y)$, $(u_j,v_x)$, and $(u_s,v_t)$,
form triples of pairwise crossing edges for $1 \leq h < i < s< j\leq p$ and $1 \leq x < t <y \leq q$,
we call the triple $(u_h,v_y)$, $(u_j,v_x)$ and $(u_s,v_t)$ \emph{maximal}. 
Note that in this maximal triple, vertex $u_i$ is between $u_h$ and $u_j$. On the other hand, if $(u_i,v_y)$, $(u_j,v_x)$, and $(u_s,v_t)$ form a triple of pairwise crossing edges with $1 \leq i < s< j\leq p$ and $1 \leq x < t <y \leq q$ such that there is no $h$ so that $(u_h,v_y)$, $(u_j,v_x)$, and $(u_s,v_t)$
also forms a triple of pairwise crossing edges  with $1 \leq h < i$, we call the triple $(u_i,v_y)$, $(u_j,v_x)$, and $(u_s,v_t)$ maximal.
We show by induction on the number of maximal triples of pairwise crossing edges  $1 \leq h < i < s< j\leq p$ and $1 \leq x < t <y \leq q$
that the number of edges of $G \neq \mathcal{S}$ is at most $\frac{9}{4}n-\frac{9}{2}$ for $n \geq 3$,
while the number of edges of $G$ is at most $1$ for $n=2$.
For the base case, assume that there is no triple of pairwise crossing edges or that $G=\mathcal{S}$.
In the former case, $G$ is quasiplanar and has at most $2n-4$ edges which is upperbounded by $\frac{9}{4}n-\frac{9}{2}$
for $n \geq 3$, while it clearly can only have one edge if $n=2$. In the latter case, i.\,e. $G=\mathcal{S}$, graph $G$ has $14$ edges.

For the induction step, we will assume that $G$ has a triple of maximal pairwise crossing edges
that connects subgraphs $G_1$ and $G_2$ as described above.
Clearly, $G_1$ and $G_2$ have fewer maximal triples of pairwise crossing edges than $G$.
As mentioned before, since $(u_i,v_{y})$ and $(u_{j},v_x)$ have two crossings from
the triple of pairwise crossing edges, all vertices between $u_i$ and $u_j$ (which are at most three),
and $v_x$ and $v_y$, respectively, are incident to a total of at most seven edges
including the edge $(u_{s},v_{t})$.
Hence, $G$ can be split into (possibly optimal) subgraphs $G_1$ and $G_2$, two isolated
vertices $u_{s}$, $v_{t}$, and possibly one more vertex $u_{s'}$ between $u_i$ and $u_j$
by removing nine edges. Note that if $u_{s'}$ exists, since every incidence to $u_s$, $v_t$ and $u_{s'}$ implies a crossing on $(u_i,v_{y})$ or $(u_{j},v_x)$, one of $u_s$, $v_t$ and $u_{s'}$  would have degree at most two. Let $w$ denote this vertex. Then 
the graph $G'$  obtained by removing $w$ has $n' = n-1$ vertices and $m'=m-2\leq \frac{9}{4}n'-\frac{9}{2}$ edges since $G' \neq \mathcal{S}$. Then, $m \leq \frac{9}{4}(n-1)-\frac{9}{2}+2 = \frac{9}{4}n-\frac{9.5}{2}$.

Let $n_1$ and $n_2$ denote the number of vertices of $G_1$ and $G_2$, respectively. Clearly, $n \geq n_1+n_2+2$. Assume first that w.l.o.g. $G_1$ is isomorphic to $\mathcal{S}$. We observe that $(u_4,v_4)$ is a planar edge of $G$, since edges $(u_2,v_4)$ and $(v_2,u_4)$ have five crossings each within $\mathcal{S}$; see \cref{fig:thomasSpecialGraph}. Then, consider the graph $G'$ obtained from $G$ by the removal of vertices $u_1$, $u_2$, $u_3$, $v_1$, $v_2$ and $v_3$. Here we consider two cases. If $G'$ is also isomorphic to $\mathcal{S}$, then $G$ contains $n=14$ vertices and $m=27 = \frac{9}{4}\cdot 14 -\frac{9}{2}$ edges. Otherwise, $G'$ has $n'=n-6$ vertices and $m'\leq \frac{9}{4}n'-\frac{9}{2}$ edges. Then, $m=m'+13\leq \frac{9}{4}(n-6)-\frac{9}{2} + 13 = \frac{9}{4}n-\frac{10}{2}$. 

Next, assume that $n_1,n_2 \geq 3$. Since we already covered the case where $G_1 = \mathcal{S}$, we may assume that $G_1$ and $G_2$ are not isomorphic to $\mathcal{S}$. Then, by induction, $G_1$ and $G_2$ have at most $\frac{9}{4}n_1-\frac{9}{2}$ and $\frac{9}{4}n_2-\frac{9}{2}$ edges, respectively. We conclude that $m \leq \frac{9}{4} (n_1+n_2) - 2\cdot \frac{9}{2}  + 9 \leq \frac{9}{4} (n_1+n_2)-9+9 = \frac{9}{4}(n_1+n_2) \leq \frac{9}{4}(n-2) = \frac{9}{4}n -\frac{9}{2}$. 

Finally, consider the case where $n_1 = 2$; see \cref{fig:5planarUpperbound2}. Because $u_i$ belongs to $G_1$,  $u_i$ can only be incident to $v_x$, $v_t$ and $v_y$. Hence, edge $(u_i,v_t)$ must exist since otherwise $u_i$ would have degree two. Symmetrically, edge $(u_s,v_x)$ is also present. Then, $u_s$ and $v_t$ can only be incident to a total of $7$ edges, if there are edges $(u_s,v_z)$, $(u_s,v_{z'})$, $(u_h,v_t)$ and $(u_{h'},v_t)$ for some $j \leq h < h' \leq p$ and $y \leq z < z' \leq q$. If $u_s$ and $v_t$ were only incident to at most $6$ edges, $G_2$ has $m_2=m-9$ edges (as it contains all edges except for those that are incident to $G_1$ or $u_s$ or $u_t$ and $n_2=n-4$ vertices. Since by induction $m_2 \leq \frac{9}{4}n_2-\frac{9}{2}$, we can conclude that $G$ has at most $m=m_2+9\leq\frac{9}{4} (n-4) + \frac{9}{2} = \frac{9}{4}n- \frac{9}{2}$ edges. Therefore, we assume in the following that $n_2 \geq 4$. 

If $n_2 = 4$, we observe that $G$ has $8$ vertices and hence has less than $\frac{9}{4}n-\frac{9}{2}$ edges, or exactly $14$ edges if it is $\mathcal{S}$.  Thus, assume that, $n_2 > 4$. We observe that all edges in $G_2$ that are incident to $u_j$ will cross edge $(u_{h'},v_t)$. Since $(u_{h'},v_t)$ already has three crossings, it follows that the degree of $u_j$ in $G_2$ is at most two. Symmetrically, the degree of $v_y$ in $G_2$ is at most two. Consider the graph $G_2'$ obtained from $G_2$ by removing $u_j$ and $v_y$; see \cref{fig:5planarUpperbound3}. Since $n_2 > 4$, $G_2'$ has  $n_2' \geq 3$ vertices. First assume that $G_2'$ is isomorphic to $\mathcal{S}$. Then, edge $(u_{j+1},v_{y+1})$ is planar, and $u_j$ and $v_y$ can only be incident to $(u_j,v_y)$, $(u_j,v_{y+1})$ and $(u_{j+1},v_y)$. Then, $G_2$ has $n_2=10$ vertices and $m_2=17 = \frac{9}{4}n_2 - \frac{11}{2}$ edges. 

Next, assume that $G_2'$ is not isomorphic to $\mathcal{S}$. We consider two cases. If $(u_j,v_y)$ is not in $G_2$ consider the graph $G_2^\ast$ obtained from $G_2$ by inserting edge $(u_j,v_y)$. Clearly, $G_2^\ast$ is $2$-layer $5$-planar and hence has at most $m_2^\ast \leq \frac{9}{4}n_2-\frac{9}{2}$ edges. Since $m_2 = m_2^\ast -  1$, it follows that $m_2 \leq \frac{9}{4}n_2-\frac{11}{2}$. So assume that  $(u_j,v_y)$ is part of $G_2$. Then, there are only three edges in $G_2$ that are not in $G_2'$. Assume that $G_2'$ has $m_2'\leq \frac{9}{4}n_2'-\frac{9}{2}$ edges. Then, $G_2$ has $m_2 = m_2' + 3 \leq  \frac{9}{4}(n_2-2)-\frac{9}{2}+3 =\frac{9}{4}n_2- \frac{12}{2}$ edges.

We conclude that $m_2 \leq \frac{9}{4}n_2-\frac{11}{2}$. Thus, $m \leq m_1 + m_2 + 9 = 1 + \frac{9}{4} n_2 -   \frac{11}{2} + 9 \leq \frac{9}{4} (n-4)-\frac{11}{2}+10 = \frac{9}{4}(n-4) + \frac{9}{2} = \frac{9}{4}n - 9 + \frac{9}{2} = \frac{9}{4}n -\frac{9}{2}$.
\end{proof}

\section{Omitted Proofs from \cref{sec:densityLargeK}}
\label{app:densityLargeK}

We first give an auxiliary lemma that will be used in the proof of \cref{thm:crossingLemma}.

\begin{restatable}{lemma}{crossingLemmaAuxiliaryLemma}
\label{lem:crossingLemmaAuxiliaryLemma}
Let $G$ be a simple $\mathcal{R}$-restricted graph with $n \geq 4$ vertices and $m$ edges.
Then, the following inequality holds for the crossing number $cr(G)$:
\begin{linenomath}
\begin{equation}
\label{eq:crossingLemmaAuxiliaryLemma}
cr(G) \geq tm - \alpha n + \beta.
\end{equation}
\end{linenomath}
\end{restatable}

\begin{proof}
Clearly, Inequality~\eqref{eq:crossingLemmaAuxiliaryLemma} holds for $m \leq \alpha_0 n - \beta_0$.
Next, assume that $m > \alpha_0 n - \beta_0$. Then there exists at least $m - (\alpha_0 n - \beta_0)$
edges in $G$ that have at least one crossing.
If $m > \alpha_1 n - \beta_1$, there exists at least $m - (\alpha_1 n - \beta_1)$
edges in $G$ that have at least two crossings.
Iteratively we obtain that, if $m > \alpha_{i-1} n - \beta_{i-1}$, there exists at least
$m - (\alpha_i n - \beta_i)$ edges in $G$ that have at least $i$ crossings.
Therefore we obtain
\begin{linenomath}
\begin{equation*}
cr(G) \geq \sum_{i=0}^{t-1} [ m - (\alpha_i n - \beta_i) ] = tm - \alpha n + \beta 
\end{equation*}
\end{linenomath}
which concludes the proof.
\end{proof}

\crossingLemma*
\begin{proof}
Consider a drawing $\Gamma$ of $G$ with $cr(G)$ crossings and let $\pi=\frac{3\alpha n}{2t m} \leq 1$.
With probability $\pi$ choose every vertex of $G$ independently and let $G_\pi$ denote the subgraph of $G$
induced by the chosen vertices, and $\Gamma_\pi$ the subdrawing of $\Gamma$ representing $G_\pi$.
Consider random variables $N_\pi$, $M_\pi$, and $C_\pi$ denoting the number of vertices and edges in $G_\pi$
and the number of crossings in $\Gamma_\pi$, respectively.
By \cref{lem:crossingLemmaAuxiliaryLemma}, it holds that $C_\pi\geq t M_\pi	- \alpha N_\pi  + \beta$.
Taking expectations on this relationship, we obtain:
\begin{linenomath}
\begin{equation*}
\pi^4cr(G) \geq t \pi^2 m	- \alpha \pi n  \implies cr(G) \geq  \frac{tm}{\pi^2}-\frac{\alpha n}{\pi^3}
\end{equation*}
\end{linenomath}
We obtain Inequality~\eqref{eq:crossingLemma} by substituting $\pi=\frac{3\alpha n}{2t m}$
into the inequality above.
\end{proof}

\kPlanarUpperbound*

\begin{proof}
If $m \leq \frac{3\alpha}{2t}n$, the proof follows immediately.
Otherwise, we obtain from \cref{thm:crossingLemma} and from the assumption that $G$ is $k$-planar:
\begin{linenomath}
$$ \frac{4 t^3}{27\alpha^2} \frac{m^3}{n^2} \leq cr(G) \leq \frac{1}{2}mk. $$
\end{linenomath}
This implies:
\begin{linenomath}
$$ m \leq \frac{3\alpha}{2t} \sqrt{\frac{3}{2t}} \sqrt{k} n $$
\end{linenomath}
which completes the proof.
\end{proof}

\kPlanarLowerbound*

\begin{proof}
We choose $p=q$. Depending on $k$, we choose a parameter $\ell$ that we will calculate later. We connect vertex $u_i$ to the $\ell$ vertices $v_{i+1}\ldots,v_{i+\ell}$. Similarly, we connect vertex $v_i$ to vertices $u_{i+1}\ldots,u_{i+\ell}$. Orient the edges from lower to higher index. Then, each vertex (except for those with indices at least $p-\ell$) has $\ell$ outgoing edges. 
Moreover, edge $(u_i,v_{i+r})$ (for $1 \leq r \leq \ell$) is only crossed by 
\begin{itemize}[--]
  \item all $\ell$ outgoing edges from vertex $v_{i+j}$ for $0 \leq j \leq r-1$; see \cref{fig:generalLowerBound1},
  \item $\ell-j$ outgoing edges from vertex $v_{i-j}$ for $1 \leq j \leq \ell-1$ which connect $v_{i-j}$
	      to vertices $u_{i+1},\ldots,u_{i+\ell-j}$; see \cref{fig:generalLowerBound2},
  \item by $r-1-j$ outgoing edges from vertex $u_{i+j}$ for $1 \leq j \leq r-2$ which connect $u_{i+j}$
	      to vertices $v_{i+j+1},\ldots,v_{r-1}$; see \cref{fig:generalLowerBound3}, and,
  \item by $\ell-r-j-1$ outgoing edges from vertex $u_{i-j}$ for $1 \leq j \leq \ell-r-2$ which connect $u_{i-j}$
	      to vertices $v_{i+r+1},\ldots,v_{i+\ell-j}$; see \cref{fig:generalLowerBound4}.
\end{itemize}
In total, for the number of edges crossing $(u_i,v_{i+r})$ we have
\begin{linenomath}
\begin{align*}
 &\mathrel{\phantom{=}} r\ell + \sum \limits_{i=1}^{\ell-1}i + \sum \limits_{i=1}^{r-2}i + \sum \limits_{i=1}^{\ell-r-2}i \\
 &= \ell r  + \frac{(\ell-1)(\ell)}{2} + \frac{(r-2)(r-1)}{2} + \frac{(\ell-r-2)(\ell-r-1)}{2}\\
 &\le \ell r + \frac{\ell^2}{2} + \frac{r^2}{2} + \frac{(\ell-r)^2}{2} = \ell^2 + r^2.
\end{align*}
\end{linenomath}
The last term is maximal for $r=\ell$, yielding at most $2 \ell^2$ crossings on $(u_i,v_{i+\ell})$.
To ensure $k$-planarity we set $2\ell^2 \leq k$ and obtain $\ell \leq \sqrt{k/2}$,
which implies that every vertex (except for those with the $\ell = \mathcal{O}(f(k))$ largest indices)
has $\ell = \sqrt{k/2}$ outgoing edges. The statement follows.
\end{proof}
}{}
\end{document}